\documentclass{llncs}

\usepackage{enumerate}
\usepackage{amssymb,amsmath}
\usepackage{enumitem}
\usepackage{verbatim}
\usepackage{hyphenat}

\usepackage{url}

\usepackage{graphicx}
\usepackage{float}
\usepackage{multicol}
\usepackage{mathtools}
\usepackage{textcomp}

\newcommand{\cent}[0]{\mbox{\textcent}}
\newcommand{\dollar}[0]{\$}

\newtheorem{fact}{Fact}

\begin{document}

\title{Uncountable realtime probabilistic classes}

\author{Maksims Dimitrijevs, Abuzer Yakary\i lmaz}
\institute{University of Latvia, Faculty of Computing \\  Rai\c na bulv\= aris 19, R\={\i}ga, LV-1586, Latvia
\\ ~ \\
\textit{md09032@lu.lv, abuzer@lu.lv}}

\maketitle

\begin{abstract}
	We investigate the minimum cases for realtime probabilistic machines that can define uncountably many languages with bounded error. We show that logarithmic space is enough for realtime PTMs on unary languages. On binary case, we follow the same result for double logarithmic space, which is tight. When replacing the worktape with some limited memories, we can follow uncountable results on unary languages for two counters.
\end{abstract}

\section{Introduction} 

When using uncountable transitions, bounded-error probabilistic and quantum models can recognize uncountably many languages \cite{ADH97,SayY14C}. It is interesting to identify the minimum resources that are sufficient to follow this result. Some of the known results \cite{SayY14C,DY16A} are as follows:
\begin{itemize}
	\item Uncountably many unary languages can be defined by poly-time double log-space probabilistic Turing machines (PTMs) and linearithmic ($ O(n \log n) $) time log-space one-way PTMs.
	\item Uncountably many $k$-ary languages ($k>1$) can be defined by poly-time constant-space quantum Turing machines, linear-time linear-space two-way probabilistic counter machines, and arbitrarily small but non-constant-space PTMs.
\end{itemize}

In this paper, we investigate \textit{realtime} probabilistic models that read the input in a streaming mode such that there is no pause on the input symbols. (This is also referred as strict realtime.) On general alphabets, it is known that bounded-error one-way PTMs cannot recognize any nonregular language in space $ o(\log \log n) $ \cite{Fr85}. Here we show that $ O(\log \log n) $-space is enough for realtime PTMs to define uncountably many languages. Therefore, this bound is tight for general alphabets. On unary alphabet, we follow the same result for $ O(\log n) $ space and we leave open whether realtime PTMs can recognize any unary nonregular languages in $ o(\log n) $ space. Lastly, we follow the same result for unary realtime probabilistic automata with counters and we show that two counters are sufficient. It is known that one counter is not enough since unary one-way probabilistic automata with one stack can recognize only regular languages with bounded error \cite{KGF97}. On the other hand, the case of two stacks is trivial since a work tape can be simulated by two stacks. We leave open to determine the minimum number of counters that use sublinear or sublogarithmic space on the counters.

In the next section, we present some background to follow the rest of the paper and then we present our results in Section \ref{sec:main-results} under two subsections. We first present the results for unary languages (Section \ref{sec:unary}), and then for general alphabet languages (Section \ref{sec:binary}).  

\section{Background}

\newcommand{\sigmastar}{\Sigma^{*}}
\newcommand{\tildesigma}{\tilde{\Sigma}}
\newcommand{\tildegamma}{\tilde{\Gamma}}
\newcommand{\tildew}{\tilde{w}}

\newcommand{\directions}{ \{ \leftarrow,\downarrow,\rightarrow \} }

We assume the reader is familiar with the basics of complexity theory and automata theory. Throughout the paper,  $ \Sigma $ not containing $\cent$ (the left end-marker) and $\dollar$ (the right end-marker) denotes the input alphabet, $ \tildesigma $ is the set $ \Sigma \cup \{ \cent,\dollar \} $, $ \Gamma $ not containing blank symbol denotes the work tape alphabet, $ \tildegamma $ is the set $ \Gamma \cup \{ \mbox{blank symbol} \} $, and $ \sigmastar $ is set of all strings obtained from the symbols in $\Sigma$ including the empty string.

Formally, a realtime PTM $ P $ is a 7-tuple
\[
P = (S,\Sigma,\Gamma,\delta,s_1,s_a,s_r ),
\]
where $ S $ is the set of finite internal states, $ s_1 \in S $ is the initial state, $ s_a \in S $ and $ s_r \in S $ ($s_a \neq s_r$) are the accepting and rejecting states, respectively, and $ \delta $ is the transition function

\[
\delta: S \times \tildesigma \times \tildegamma \times S \times \tildegamma \times \directions \rightarrow [0,1]
\] 
that governs the behaviour of $P $ as follows: When $ P $ is in state $ s \in S $, reads symbol $ \sigma \in \tildesigma $ on the input tape, and reads symbol $ \gamma \in \tildegamma $ on the work tape, it enters state $ s' \in S $, writes $ \gamma' \in \tildegamma $ on the cell under the work tape head, and then the work tape head is updated with respect to $ d \in \directions $ with probability
\[
	\delta(s,\sigma,\gamma,s',\gamma',d),
\]
where ``$ \leftarrow $'' (``$\downarrow$'' and ``$\rightarrow$'') means the head is moved one cell to the left (the head does not move and the head is moved one cell to the right). Note that input head can only perform ``$\rightarrow$'' moves. To be well-formed PTM, the following condition must be satisfied: for each triple $(s,\sigma,\gamma) \in S \times \tildesigma \times \tildegamma $,
\[
	\sum_{s' \in S,\gamma' \in \tildegamma,d \in \directions} \delta(s,\sigma,\gamma,s',\gamma',d) = 1.
\]

The computation starts in state $ s_1 $, and any given input, say $ w \in \Sigma^* $, is read as $ \cent w \dollar $ from the left to the right symbol by symbol, and the computation is terminated and the given input is accepted (rejected) if $ P $ enters $ s_a $ ($s_r$). It must be guaranteed that the machine enters a halting state after reading $\dollar$.

The space used by $ P $ on a given input is the number of all cells visited on the work tape during the computation with some non-zero probability. The machine $ P $ is called to be $ O(s(n)) $ space bounded machine if it always uses $ O(s(n)) $ on any input with length $n$. 

If (realtime) $ P $ is allowed to spend more than one step on an input symbol, then it is called one-way. Formally, its transition function is extended by the move of the input head with $ \{\downarrow,\rightarrow\} $ in each transition, and then, the well-formed condition is updated accordingly.

Moreover, any PTM without work tape is called probabilistic finite automaton (PFA).

A counter is a special type of memory containing only the integers. Its value is set to zero at the beginning. During the computation, its status (whether its value is zero or not) can be read similar to reading blank symbol or non-blank symbol on the work tape, and then its value is incremented or decremented by 1 or not changed similar to the position update of the work head. (A counter can be seen as a unary stack.) 

A realtime probabilistic automaton with $ k $ counters (P$k$CA) is a realtime PTM having $ k $ counters instead of a working tape. In each step, instead of reading the symbol under the work tape head, it checks the statuses of all counters; and then, it updates the value of each counter by a value from $ \{-1,0,1\} $ instead of updating the content of the work tape. 

The language $L$ is said to be recognized by a PTM with error bound $\epsilon$ ($0 \leq \epsilon < 1/2$) if every member of $ L $ is accepted with probability at least $1-\epsilon$ and every non-member of $L$ ($w \notin L$) is accepted with probability not exceeding $\epsilon$.

We denote the set of integers $ \mathbb{Z} $ and the set of positive integers $ \mathbb{Z}^+ $. The set $ \mathcal{I} = \{ I \mid I \subseteq \mathbb{Z^+} \} $ is the set of all subsets of positive integers and so it is an uncountable set (the cardinality is $ \aleph_1 $) like the set of real numbers ($ \mathbb{R} $). The cardinality of $ \mathbb{Z} $ or $ \mathbb{Z^+} $ is $ \aleph_0 $ (countably many). 

The membership of each positive integer in any $ I \in \mathcal{I} $ can be represented as a binary probability value:
\[
p_I = 0.x_1 0 1 x_2 0 1 x_3 0 1 \cdots x_i 0 1 \cdots,~~~~ x_i = 1 \leftrightarrow i \in I.
\]  

\section{Our results}
\label{sec:main-results}
In our proof we use a fact presented in our previous paper \cite{DY16A}.

\begin{fact}
	\label{fact:DY16A}
	\cite{DY16A}
	Let $x=x_1 x_2 x_3 \cdots$ be an infinite binary sequence. If a biased coin lands on head with probability  $p = 0. x_1 0 1 x_2 0 1 x_3 0 1 \cdots$, then the value $x_k$ can be determined with probability at least $\frac{3}{4}$ after $64^k$ coin tosses.
\end{fact}

The proof of this fact involves the analysis of probabilistic distributions for the number of heads after tossing $64^k$ coins that land on the head with probability $p$. The $(3 \cdot k+3)$-th bit from the right in obtained number of heads is equal to $x_k$ with probability at least $\frac{3}{4}$.

\subsection{Unary languages}
\label{sec:unary}

In \cite{YS13B}, it was shown that realtime deterministic Turing machines (DTMs) can recognize unary nonregular languages in $O(\log n)$ space. By adopting the technique given there, we can show that bounded-error realtime PTMs can recognize uncountably many unary languages.

\newcommand{\ulog}{\mathtt{ULOG}}
\newcommand{\ulogI}{\mathtt{ULOG(I)}}

\begin{theorem}
	\label{thm:log-uPTM}
	Bounded-error realtime unary PTMs can recognize uncountably many languages in $O(\log n)$ space.
\end{theorem}
\begin{proof}
	We start with defining a unary nonregular language that can be recognized by bounded-error log-space realtime PTMs:
	\[
		\ulog = \{ 0^{k_i} \mid k_1=64 \cdot 28 \mbox{ and } k_i=k_{i-1} + 64^{i} \cdot (18i+10) \mbox{ for } i > 1 \},
	\]
	where each member is defined recursively. Since it is not a periodic language, $ \ulog $ is nonregular.
	
	For any $ I \in \mathcal{I} $, we define the following language:
	\[
		\ulogI = \{ a^{k_i} \mid a^{k_i} \in \ulog \mbox{ for } i \geq 1 \mbox{ and } i \in I \}.
	\]
	
	We describe a bounded-error log-space PTM for $ \ulogI $, say $P_I$. Then, we can follow the proof since there is a bijection (one-to-one and onto) between $ I \in \mathcal{I} $ and $ \ulogI $ and $ \mathcal{I} $ is an uncountable set.	
	
	The PTM $ P_I $ uses a coin landing on head with probability 
	\[
		p_I=0. x_1 0 1 x_2 0 1 x_3 0 1 \cdots x_i 0 1 \cdots,
	\]
	where $ x_i = 1 $ if and only if $ i \in I $. The aim of $ P_I $ is iteratively finding the values of $ x_1, x_2, \ldots $ with high probability. If all input is read before reaching a decision on one of these values, then the input is always rejected. 
	
	During the computation, $ P_I $ uses two binary counters on the work tape. At the beginning, the iteration number is one, $ i = 1 $. The machine initializes the work tape as ``\#000000\#000000\#'' by reading 15 ($=9 \cdot 1 + 5 + 1$) symbols from the input (after 15-th symbol the working tape head is placed on the first zero to the left from the third $\#$). We name the separator symbols $ \# $s for the counters as the first, second, and third ones from left to the right. The first (second) counter is kept between the last (first) two $ \# $s.  
	
	By using the first counter, the machine counts up to $ 64^{i} $ and so meanwhile also tosses $ 64^{i} $ coins. By using the second counter, it counts the number of heads. The value of each counter can be easily increased by 1 when the working tape head passes on the counters from right to left once. Thus, when the working tape head is on the third $\#$, it goes to the first $\#$, and meanwhile increases the value of the first counter by 1, then tosses its coin, and, if it is a head, it also increases the value of the second counter. After tossing $ 64^{i} $ coins, the machine uses the leftmost value of the second counter as its answer for $ x_i $. Once this decision is read from the work tape and immediately after the working tape head is placed on the first $ \# $, the current iteration is finished. If (i) an iteration is finished, (ii) there is no more symbol remaining to be read from the input, and (iii) the decision is positive, then the input is accepted, which is the single condition to accept the input. After an iteration is finished, the next one starts and each counter is initialized appropriately and then the same procedure is repeated as long as there are some input symbols to be read.
	
	Since the input is read in realtime mode, the number of computational steps is equal to the length of the input plus two (the end-markers). Now, we provide the details of each iteration step so that we can identify which strings are accepted by $ P_I $.
	
	At the beginning of the $ i $-th iteration, the working tape head is placed in the first $ \# $ and the contents of the counters are as follows:
	\[
		\# \underbrace{0 \cdots 0 }_{3(i-1)+3} \# \underbrace{0 \cdots 0 }_{6(i-1)} \#. 
	\]
	By reading $ 9i+5 +1$ symbols from the input, the counters are initialized for the current iteration as 
	\[
		\# \underbrace{0 \cdots 0 }_{3i+3} \# \underbrace{0 \cdots 0 }_{6i} \# 
	\]
	by shifting the second and third $ \# $s to 3 and 9 amounts of cells to the right (after initialization the working head is placed on the first zero to the left from the third $\#$).
	
	After the initialization of the counters, the working head goes to the first $ \# $ and then comes back on the third $ \# $ $ 64^{i} -1$ times. In each pass from right to left, the first counter is increased by 1, the coin is flipped, and then the second counter is increased by 1 if the result is head. When all digits of the first counter are 1, which means the number of passes reaches $ 64^{i} - 1 $, the working tape head makes its last pass from the third $ \# $ to the first $ \# $. During the last pass, $ P_I $ flips the coin once more and then determines the leftmost digit of the second counter. Meanwhile, it also sets both counters to zeros. 
	
	By also considering the initialization step, $ P_I $ makes $ 64^i $ passes starting from the first $ \# $. So, the total number of steps is $ 64^i \cdot 2 \cdot (9i+5) $ during the $ i $-th iteration. One can easily verify that this is valid also for the case of $ i = 1 $.
	
	Therefore, $ P_I $ can deterministically detect the $ i $-th shortest member of $ \ulog $ after reading $ k_i $ symbols, where $ k_1 = 64 \cdot (28)$ and $ k_i = k_{i-1} + 64^i \cdot (18i+10) $ for $ i > 1 $. Then, by using Fact \ref{fact:DY16A}, we can follow that $ P_I $ recognizes $ \ulogI $ with error bound $ \frac{1}{4} $. 
		\qed
\end{proof}

It is known that bounded-error unary one-way PFAs with a single stack cannot recognize any nonregular language \cite{KGF97}. Therefore, we can check the case of having two stacks.

\begin{corollary}
	Bounded-error unary realtime PFA with two stacks using logarithmic amount of space can recognize uncountably many languages.
\end{corollary}
\begin{proof}
	It is a well-known fact that two stacks can easily simulate a worktape of a TMs without any delay on the running time. Therefore, by using Theorem \ref{thm:log-uPTM}, we can follow the result in a straightforward way. \qed
\end{proof} 

It is possible to replace stacks with counters by losing the space efficiency. We start with four counters.

\begin{theorem}
	\label{thm:4PCA}
	Bounded-error realtime unary P4CAs can recognize uncountably many languages.
\end{theorem}
\begin{proof}
	We start with describing a realtime P4CA, say $ P_I $, that can use a coin landing head with probability $ p_I $ for an $ I \in \mathcal{I} $. Let $ C_i $ ($ 1 \leq i \leq 4 $) represent the values of counters.
	
	The automaton $P_I$ executes an iterative algorithm. We use $ m $ to denote the iteration steps. At the beginning, $ m = 1 $. In each iteration, $ 64^m $ coin tosses are performed. The details are as follows:
	\begin{itemize}
		\item Set $ C_1 = 64^m $ and $ C_2 = 4 \cdot 8^m $.
		\item Perform $ C_1 $ coin flips and meanwhile increase/decrease the values of $ C_2 $ and $ C_3 $ by 1. If the coin flip result is head, one of the counters is increased by 1 and the other one is decreased by 1. When one of them hits zero, update strategy is changed. Since $ C_3 $ is zero at the beginning, the first strategy is decreasing the value of $ C_2 $ and increasing the value of $ C_3 $. Thus, after each $4 \cdot 8^m $ heads, the update strategy on the counters is changed.
		\item When $ C_1 $ hits zero, $ C_2 $ and $ C_3 $ are equal to $ X $ and $ 4\cdot 8^m - X $, and, the automaton makes its decision on $ x_m $. If the latest strategy is decreasing the value of $ C_3 $ or $ C_2 = 0 $, then $ x_m $ is determined as 1. Otherwise, it is determined as 0.
	\end{itemize}
	
	The described algorithm is similar to the one that is used in the proof of Theorem \ref{thm:log-uPTM}. Here changing the update strategy between $ C_2 $ and $ C_3 $ refers to the change of bit $ x_m $, which is changed after each $4 \cdot 8^m $ heads: it is 0 initially and then changed as $ 1, 0, 1, \ldots $.
		
	At the end of the $ m $-th iteration, we have $C_1 = 0 $, $ C_2 = X $, and $ C_3 = 4 \cdot 8^m -X $. We initialize $ (m+1) $-th iteration as follows:
	\begin{itemize}
		\item By using $ C_2 $ and $C_3$, we can set $ C_1 = 2X + 2 (4 \cdot 8^m - X) = 8^{m+1} $. Now $ C_2 = C_3 = C_4 = 0 $.
		\item Set $ C_2 = C_3 = 8^{m+1} $ by setting $ C_1 = 0 $. Then, in a loop, until $ C_2 $ hits zero: decrease value of $ C_2 $ by 1, then transfer $ C_3 $ to $ C_4$ (or $ C_4 $ to $ C_3 $ if at the beginning of loop's iteration $ C_3 = 0$) and meanwhile add $ 8^{m+1} $ to $ C_1 $.
		\item $ C_1 = 8^{m+1} (8^{m+1}) = 64 ^{m+1} $, $ C_2 = 0 $, $ C_3 = 8^{m+1} $, $ C_4 = 0 $. Then set $ C_2 = 4 \cdot 8^{m+1} $ by setting $ C_3 = 0 $. 
	\end{itemize}
	After initializing, we execute the coin-flip procedure. Each iteration finalizes after coin-flip procedure.
	
	The input is accepted if there is no more input symbol to be read exactly at the end of an iteration, say $m$-th, and $ x_m $ is guessed as $ 1 $. Otherwise, the input is always rejected.
	
	The coin tosses part is performed in $64^m$ steps. The initialization part for $m$-th iteration is performed in $ 8^m + 8^m + 64^m + 4 \cdot 8^m = 64^m + 6 \cdot 8^m $ steps, where $m>1$. The initialization part for $ m = 1 $ is performed in 64 steps.
	
	Based on this analysis, we can easily formulate the language recognized by $ P_I $, which is subset of the following language
	\[
	\mathtt{UP4CA} = \{ 0^{k_i} \mid k_1=128 \mbox{ and } k_i=k_{i-1} + 6 \cdot 8^i + 2 \cdot 64^i \mbox{ for } i > 1 \}.
	\]	
	For any $ I \in \mathcal{I} $, the realtime P4CA $ P_I $ can recognize the language
	\[
	\mathtt{UP4CA(I)} = \{ a^{k_i} \mid a^{k_i} \in \mathtt{U4PCA} \mbox{ for } i \geq 1 \mbox{ and } i \in I \}
	\]
	with bounded error. The automaton $ P_I $ iteratively determines the values of $ x_1, x_2, \ldots $ with high probability and the number of steps for each iteration corresponds with the members of $\mathtt{U4PCA}$.
	
	Since $ \mathcal{I} $ is an uncountable set and there is a bijection between  $ I \in \mathcal{I} $ and $ \mathtt{UP4CA(I)} $, realtime P4CAs can recognize uncountably many unary languages with bounded error. 
	\qed
\end{proof}

We can establish a similar result also for realtime P2CAs. For this purpose, we can use the well-known simulating technique of $ k $ counters by 2 counters.

\begin{theorem}
	\label{thm:2PCA}
	Bounded-error unary realtime P2CAs can recognize uncountably many languages.
\end{theorem}
\begin{proof}
	Let $ P_I $ be the realtime P4CA described above and $ \mathtt{UP4CA(I)} $ be the language recognized by it. 
		Due to the realtime reading mode, the unary inputs to $ P_I $ can also be seen as the time steps. For example, $ P_I $ can be seen as a machine without any input but still making its transition after each time step. Thus, after each step it can be either in an accepting case or a rejecting case.
	
	It is a well-known fact that two counters can simulate any number of counters with big slowdown \cite{Min67}. The values of $ k $ counters, say $ c_1,c_2,\ldots,c_k $, can be stored on a counter as 
	\[
		p_1^{c_1} \cdot p_2^{c_2} \cdot \cdots p_k^{c_k},
	\]
	where $ p_1,\ldots,p_k $ are some prime numbers. Then, by the help of the second counter and the internal states, it can be easily detected and stored the status of each simulated counters, and then all updates on the simulated counters are reflected one by one.
	
	Thus, by fixing the above simulation, we can easily simulate $ P_I $ by a P2CA, say $ P_I' $. Then, $ P_I' $ recognizes a language with bounded error, say $ \mathtt{UP2CA(I)} $. 
	
	It is easy to see that there is a bijection between 
	\[
		 \{ \mathtt{UP4CA(I)} \mid I \in \mathcal{I} \} \mbox{ and } \{ \mathtt{UP2CA(I)} \mid I \in \mathcal{I} \}, 
	\]
	and so realtime P2CAs also recognize uncountably many languages with bounded error. Remark that for each member of $ \mathtt{UP4CA(I)} $, the corresponding member of $ \mathtt{UP2CA(I)} $ is much longer.
	\qed
\end{proof}

\subsection{Generic alphabet languages}
\label{sec:binary}

Here, we focus on non-unary alphabets and establish our result for double logarithmic space. For this purpose, we use a fact given by Freivalds in \cite{Fre83}. 

\begin{fact}
	\label{fact:Fre83}
	Let $P_1 (n)$ be the number of primes not exceeding $2^{ \lceil log_2 n \rceil }$, $P_2 (l,N',N'')$ be the number of primes not exceeding $2^{ \lceil log_2 l \rceil }$ and dividing $|N'-N''|$, and $P_3 (l,n)$ be the maximum of $P_2 (l,N',N'')$ over all $N'<2^n$, $N'' \leq 2^n$, $N' \neq N''$.
		Then, for any $\epsilon > 0$, there is a natural number $c$ such that $\lim_{n\to\infty} \frac{P_3 (cn,n)}{P_1 (cn)} < \epsilon$.  
\end{fact}

\newcommand{\loglog}{\mathtt{LOGLOG}}
\newcommand{\loglogI}{\mathtt{LOGLOG(I)}}

Let $bin(i)$ denote the unique binary representation of $i>0$ that always starts with digit 1. The language $ \loglog $ is composed by the strings
\[
	bin(1) 2 bin(2) 2 bin(3) 2 ...2 bin(s) 4,
\] 
where $ |bin(s)| = 64^k $ for some positive integer $k$. For any $ I \in \mathcal{I} $, we define language $ \loglogI = \{  w \mid w \in \loglog \mbox{ and } k \in I \} $.

\begin{fact}
	\label{fact:dist-of-primes}
	Denote by $\pi (x)$ the number of primes not exceeding $x$. The Prime Number Theorem states that $\lim_{x\to\infty} \frac{\pi (x)}{ x / \ln x} = 1 $ \cite{Chandrasekharan1968}.
\end{fact}

\begin{theorem}
	\label{thm:onewayPTM}
	Bounded--error one--way PTMs can recognize uncountably many languages in $O(\log \log n)$ space.
\end{theorem}
\begin{proof}
	By modifying the one-way algorithm given in \cite{Fre83}, we present a PTM, say $P_{c,I}$, shortly $P$, for language $ \loglogI $ for $I \in \mathcal{I}$ and for a specific $c$ that determines the error bound. $P$ performs different checks by using the separate parts of the work tape.
	
	For each $ i $, $P$ keeps two registers storing $ m = |bin(i)| $ and $ m_0 = |bin(i-1)| $. After reading $bin(i)$, $P$  checks: if $m=m_0$ or ($m=m_0+1$ and $bin(i-1)$ contained only ones), then $P$ continues. Otherwise, $P$ rejects the input.
	
	For each $bin(i)$, $P$ generates a random number of $|m| \cdot c$ bits and tests it for primality. If the generated number is not prime, the same procedure is repeated.
	Due to Fact \ref{fact:dist-of-primes}, we can follow that the probability of picking a prime number of $|m| \cdot c$ bits is $ \theta (\frac{1}{|m| \cdot c})$. Therefore, the expected time of finding a prime number is $ O(|m| \cdot c) $.
	Assume that the generated prime number is $r_i$. For each $bin(i)$, $P$ calculates $bin(i) \mod r_i$ and $bin(i+1) \mod r_i$. If $(bin(i) \mod r_i)+1 \neq bin(i+1) \mod r_i$, $P$ rejects the input. Otherwise, the computation continues.
	
	After reading ``4'', $P$ checks whether $m=64^k$ for some integer $k>0$. If so, $m$ is written on the tape as $1(000000)^k$. If $m \neq 64^k$, then the input is rejected.
	
	If all previous checks are successful, $P$ tosses $ 64^k $ coins and meanwhile calculates the number of heads $ \mod (8 \cdot 8^k) $, say $ C $. If after all coin tosses, the leftmost bit of $C$ is 1, then the input is accepted, otherwise it is rejected.
	
	The PTM $P$ reaches symbol ``4'' without rejecting with probability 1 if the input belongs to $\loglog$, and it rejects the input before reaching ``4'' with probability at least $1-\epsilon$ if the input is not in $\loglog$ due to Fact \ref{fact:Fre83}. Due to Fact \ref{fact:DY16A} the membership of $k \in I$ for $\loglogI$ will be computed with probability at least $\frac{3}{4}$. Therefore language $\loglogI$ is recognized correctly with probability at least $(1-\epsilon) \cdot \frac{3}{4}$, which can be arbitrarily close to $\frac{3}{4} $ by picking a suitable $c$.
	
	The space used on the work tape is linear in the length of the counter for $|bin(i)|$. The value of $bin(i)$ is logarithmic to the length of input word, and so the length of the counter is double logarithmic to the input length. Therefore, the space used is in $O(\log \log n)$ throughout the computation.
	\qed
\end{proof}

Let $ L \subseteq \Sigma^* $ be a language recognized by a one-way DTM, say $D$, and $ \sigma $ be a symbol not in $ \Sigma $. We can execute $ D $ in realtime reading mode on the inputs defined on $ \Sigma \cup \{\sigma\} $ as follows \cite{YS13B}: For each original ``wait'' move on a symbol from $ \Sigma $, the machine expects to read symbol $\sigma $. If it reads something else or there is no more input symbol, then the input is rejected. If there is more than expected $ \sigma $ symbol, then again the input is rejected. Thus, we can say that this modified machine recognizes a language $ L' $ and there is a bijection between $ L $ and $L'$. Moreover, the space and time bounds for both machines are the same.

The question is whether we can apply a similar idea for one-way PTM given above in order to get a realtime PTM. A DTM follows a single path during its computation and so the aforementioned bijection can be created in a straightforward way. On the other hand, PTMs can follow different paths with different lengths in each run. So, in order to follow a similar bijection, we need some modifications. The main modification is necessary for the task of picking the prime numbers. Except this task, the other ones can be executed with the same number of steps (remember the algorithms in the previous subsection) in every execution of the machine. 

Now, we modify PTM $P_{c,I}$ in order to guarantee that each computation path uses the same amount of time steps on the same input. We represent the new PTM as $ P'_{c,I} $ or shortly $P'$.

 The PTM $P'$ uses some registers on the work tape separated by ``$\#$'': 
 \[
 	\#1st\#2nd\# \cdots \#last\#.
 \]
\begin{itemize}
	\item The 1st register keeps both the lengths of the counters $m$ and $m_0$. If $m= x_1 x_2 x_3 \cdots$ and $m_0= y_1 y_2 y_3 \cdots$, then the register keeps the values in the following way: $x_1 y_1 x_2 y_2 x_3 y_3 \cdots$. After reading symbol ``2'' it is easy to compare $m$ and $m_0$ bit by bit with a single pass.
	\item The 2nd register keeps the number of heads for the coin-tosses, based on which the bit $x_k$ is determined. It is set to $\lceil |m|/2 \rceil+2$ zeros before any coin-toss procedure and it is updated accordingly when the value of $m$ is changed.
	\item The 3rd register keeps the track of attempts to generate prime number, it has $|m| \cdot c$ bits.
	\item The 4th and 5th registers keep the prime numbers with some auxiliary numbers. Each register has $|m| \cdot c \cdot 2$ bits. If the (candidate) prime number is $r = r_1 r_2 r_3 \cdots$ and the auxiliary number is $q = q_1 q_2 q_3 \cdots$, then the register keeps both of them as $r_1 q_1 r_2 q_2 r_3 q_3 \cdots$. The machine uses $r$ to store the prime number that is being checked or computed, and $q$ is used to help to perform tasks with $r$ like storing number modulo $r$ and comparing and copying numbers. For each $j>0$, the machine uses 4th and 5th registers to work with prime numbers and then checks the correctness of the candidates for $bin(2 \cdot j-1)$ and $bin(2 \cdot j)$.
	\item The 6th and 7th registers are the same as 4th and 5th registers, respectively. Only they are responsible for the correctness of the candidates for $bin(2 \cdot j)$ and $bin(2 \cdot j+1)$.
	\item The 8th register has a number to keep track of total number of subtractions performed while checking the divisibility of $r$ by $d$. It has $|m| \cdot c$ bits.
	\item The 9th register has twice of $\lceil |m|/2 \rceil \cdot c$ bits to keep numbers $d$ and $h$ (each is $\lceil |m|/2 \rceil \cdot c$ bits). If $d = d_1 d_2 d_3 \cdots$ and $h = h_1 h_2 h_3 \cdots$, then the register keeps them as $d_1 h_1 d_2 h_2 d_3 h_3 \cdots$. Both numbers are used to check whether the generated number $r$ is prime. The machine uses $d$ to check whether $d$ does not divide $r$, such check is performed for different values $d$. The check is performed by making subtractions. The value of $d$ is subtracted from $r$ multiple times. For this operation, the machine uses $h$ as auxiliary number.
\end{itemize}

Each member of $\loglogI$ has parts $bin(i)$ at least up to $bin(2^{63})$. $P'$ deterministically checks input up to $bin(2^{63})$ and prepares the work tape with 9 registers.

Now, we describe the steps of picking prime numbers.

For number $bin(i)$, the prime number is generated in $(6 - 2 \cdot (i \mod 2))$-th register. The number $r$ is generated by using $|m| \cdot c$ random bits (bit by bit). 
After this, the primality check is performed. For this purpose, the machine checks whether $r$ is divided by any natural number between 2 and $2^{\lceil |m|/2 \rceil \cdot c}-1$, where $2^{\lceil |m|/2 \rceil \cdot c} > sqrt(r)$ because $r < 2^{|m| \cdot c}$. Each candidate natural number is denoted by $d$ below. Remark that the number of $d$s does not depend on $r$ and so for any candidate prime number, the primary test procedure takes the same number of steps.

To begin the check of divisibility of $r$ by $d$, the value of $r$ is copied to $q$ bit after bit, and the value of $d$ is copied to $h$ bit after bit. The 8th register is initialized with zeros before check for pair $r$ and $d$. Then, $2^{|m| \cdot c}$ iterations are performed. In each iteration, the values of $q$ and $h$ are decreased by 1, the value of 8th register is increased by 1. If only $h$ reaches zero, $d$ is again copied into $h$ and the machine continues to perform iterations. When $ q $ reaches zero, if $ h $ reaches zero at the same time, the machine concludes that $ r $ is not a prime number, otherwise, $r$ is not divisible by $d$. After that, $P'$ continues to perform the iterations but without changing $q$ and checks of value of $q$ until the value of the 8th register reaches $2^{|m| \cdot c}$. Then, $P'$ repeats the procedure for the next $d$.

If $r$ is not divisible by any of these $d$s, then the procedure of finding prime is terminated successfully since $r$ is prime, otherwise, the machine continues with the next prime candidate number since $ r $ is not a prime number.

The 3rd register counts the number of attempts to generate a prime number. It is initialized with zeros and is increased by one after each try. If $P'$ finds a prime number before 3rd register reaches $2^{|m| \cdot c}$, $P'$ continues performing the algorithm until the register reaches $2^{|m| \cdot c}$ by fixing the candidate with the already found prime number. If the register reaches value $2^{|m| \cdot c}$ (all bits become zeros) and $P'$ fails to generate a prime number, $P'$ uses the last generated $ r $ for the modular check for pair $bin(i)$ and $bin(i+1)$. $P'$ performs each try to generate (or process already generated) prime number in equal number of steps. For any $bin(i)$ $P'$ performs exactly $2^{|m| \cdot c}$ such operations.

After finding and checking prime $r$, the machine copies $r$ into $(7 - 2 \cdot (i \mod 2))$-th register bit by bit. To perform this operation, the machine sets $q$ to zeros in both registers, copies the bits of $r$ one by one, and marks the copied bit by setting the next bit in $q$ to one.

Now, we describe how the machine calculates the value $bin(i) \mod r$. At the beginning, the register keeps $r$ and zeros for $q$. Assume that $bin(i) = i_1 i_2 \cdots i_m$. When the machine reads $i_j$, the value of $q$ is multiplied by 2 and increased by $i_j$. Therefore, all bits of $q$ are shifted to left by one position, and the machine puts value $i_j$ in leftmost bit. If, after this operation, $q \geq r$, then $r$ is subtracted from $q$. Because both values are interleaved, it is easy to subtract $r$ from $q$ in one pass. In the case when $q < r$ the machine performs one pass through registers without changing the values. This ensures that each iteration for $i_j$ is performed in equal number of steps. The machine performs the calculation while reading $bin(i)$ for the 5th and the 6th registers if $i \mod 2 = 0$, and for the 4th and the 7th registers otherwise.

After these, the machine compares the values of two modules: the 4th and the 5th registers if $i \mod 2 = 0$; the 6th and the 7th registers otherwise. This time machine sets $r$ in both registers to zeros and marks compared bits of $q$'s by setting bits in $r$ to one.

If $ r $ in modular check is not prime, $P'$ cannot guarantee that incorrect pair $bin(i)$ and $bin(i+1)$ will be rejected with probability at least $1-\epsilon$. The probability not to generate a prime number of $|m| \cdot c$ random bits in $2^{|m| \cdot c}$ tries does not exceed ${(1 - \frac{1}{|m| \cdot c})}^{2^{|m| \cdot c}}$ because of Fact \ref{fact:dist-of-primes}. Note that $\lim_{n\to\infty} {{(1 - \frac{1}{n})}^n} = \frac{1}{e}$, therefore $\lim_{m\to\infty} {(1 - \frac{1}{|m| \cdot c})}^{2^{|m| \cdot c}} = \lim_{m\to\infty} {\frac{1}{e}}^{\frac{2^{|m| \cdot c}}{|m| \cdot c}} = 0$. 
The smallest $|m|$ for which a prime number is generated is 7. By picking a suitable $c$, the value ${(1 - \frac{1}{7 \cdot c})}^{2^{7 \cdot c}} = \epsilon_0$ can be arbitrarily close to zero. 
For each $i>0$, checking the equality of $bin(i)$ and $bin(i+1)$ by using the generated prime number is performed independently. Therefore, any incorrect pair is accepted with probability at most $\epsilon$ due to Fact \ref{fact:Fre83}. Since $P'$ can fail to generate a prime number, this probability is increased to at most $\epsilon + \epsilon_0 - \epsilon \cdot \epsilon_0$. If the input belongs to $\loglogI$, $P'$ is guaranteed to not reject the input before reaching ``4'' on input tape. If at least one pair $bin(i)$ and $bin(i+1)$ is inacceptable, then $P'$ rejects input right after checking this pair with probability at least $1 - \epsilon - \epsilon_0 \cdot (1 - \epsilon)$. Therefore, the error remains bounded.

The other parts of the algorithm are executed with the same number of steps in every execution of $P'$.

\begin{theorem}
	Bounded--error realtime PTMs can recognize uncountably many languages in $O(\log \log n)$ space.
\end{theorem}

\begin{proof}
	We can obtain a realtime algorithm from $P'$, say $ R_{c,I} $ or shortly $R$, by using aforementioned technique borrowed from \cite{YS13B}. Let $\loglogI'$ be the language recognized by $ R $. Then, the language $\loglogI'$ differs from the language $\loglogI$ with the presence of symbols ``3'': for each ``wait'' move on ``0'', ``1'', ``2'' or ``4'' by $P'$, $ R $ expects to read one symbol of ``3''. If $ R $ fails to read a symbol of ``3'' when it is expected, the input is rejected. 
		
	PTM $P'$ recognizes $\loglogI$ in $O(\log \log n)$ space, therefore, realtime machine $R$ recognizes $\loglogI'$ in $O(\log \log n)$ space and there is a bijection between $\loglogI$ and $\loglogI'$.
	\qed
\end{proof}

In \cite{Fre83} Freivalds has proven that only regular languges can be recognized with one-way PTM in $o(\log \log n)$ space and with probability $p > \frac{1}{2}$. Therefore, the presented space bound is tight.

\section*{Acknowledgments}
Dimitrijevs is partially supported by University of Latvia project AAP2016/B032 ``Innovative information technologies''. Yakary{\i}lmaz is partially supported by ERC Advanced Grant MQC. We thank to the reviewers for their helpful comments.

\bibliographystyle{splncs03}
\bibliography{tcs}

\end{document}